\documentclass[final,5p,twocolumn]{elsarticle}

\usepackage{graphicx}          % Include this line if your
\usepackage[dvips]{epsfig}    % or this line, depending on which
\usepackage{xcolor}
\usepackage{amssymb}
\usepackage{amsthm}
\usepackage{amsmath}
\usepackage{algorithmic}

\journal{Mechanical Systems and Signal Processing}

\begin{document}

\begin{frontmatter}
%\runtitle{AAA BBB CCC}

% Running title for regular

\title{One-parameter robust global frequency estimator for
slowly varying amplitude and noisy oscillations} % Title, preferably not more

\author{Michael Ruderman}
\ead{michael.ruderman@uia.no}

\address{University of Agder, 4604-Norway}

\begin{abstract}                          % Abstract of not more than 200 words.
Robust online estimation of oscillation frequency belongs to
classical problems of system identification and adaptive control.
The given harmonic signal can be noisy and with varying amplitude
at the same time, as in the case of damped vibrations. A novel
robust frequency-estimation algorithm is proposed here, motivated
by the existing globally convergent frequency estimator. The
advantage of the proposed estimator is in requiring one design
parameter only and being robust against measurement noise and
initial conditions. The proven global convergence also allows for
slowly varying amplitudes, which is useful for applications with
damped oscillations or additionally shaped harmonic signals. The
proposed analysis is simple and relies on an averaging theory of
the periodic signals. Our results show an exponential convergence
rate, which depends, analytically, on the sought frequency,
adaptation gain and oscillation amplitude. Numerical and
experimental examples demonstrate the robustness and efficiency of
the proposed estimator for signals with slowly varying amplitude
and noise.
\end{abstract}


\end{frontmatter}

\newtheorem{thm}{Theorem}
\newdefinition{rmk}{Remark}
\newtheorem{lem}[thm]{Lemma}
%\newtheorem{defn}{Definition}
%\newtheorem{clr}{Corollary}
%\newdefinition{exmp}{Example}
%\newdefinition{prop}{Proposition}
%\newproof{pf}{Proof}

%%%%%%%%%%%%%%%%%%%%%%%%%%%%%%%%%%%%%%%%%%%%%%%%%%%%%%%%%%%%%%%%%%%%%%%%%%%%%%%%
\section{Introductory note}
\label{sec:1}

A common problem associated with online estimation of the unknown
frequency of harmonic signals has been studied in multiple works
(see, e.g.,
\cite{bodson1997,hsu1999,marino2003,mojiri2004,vedyakov2017} and
references in the latter). Some differing approaches rely on
extended observer or Kalman filter principles (see, e.g.,
\cite{bittanti2000}). Other approaches (particularly those
accommodated in power electronics applications) use so-called
phase-locked-loop (PLL) algorithms (see, e.g., \cite{karimi2004}).
Frequency estimation can be seen to be one of the fundamental
questions in systems and signals theory, and it has multiple
practical mechanical and electrical applications. For instance, it
can be found in the following: power system converters and
controllers (e.g., \cite{karimi2004}); active control of sound and
vibrations (e.g., \cite{fuller1995}); rotary machines, like in
magnetic bearings (e.g., \cite{herzog1996}); drives with
eccentricities (e.g., \cite{DeWit2000}); and motion control with
periodic and vibrational disturbances of, e.g., disk drives
\cite{sacks1996} or suspensions \cite{landau2005}, to mention just
a few.

In several application scenarios, like in the case of damped
vibrations, the time-varying amplitude of a harmonic signal is,
however, the most challenging factor for robust and sufficiently
fast estimation of the unknown frequency. Besides, measurement and
process noise can further degrade those estimation approaches for
which the convergence is theoretically proven, but which can
suffer through sensitivity to real measured (physical) data,
making their applicability questionable in terms of robust
convergence. Among the numerous existing frequency-estimation
methods, the globally convergent estimator, introduced in
\cite{hsu1999}, appears promising due to its structural simplicity
and low number of design parameters. This paper focuses on the
globally convergent frequency estimator (cf.
\cite{hsu1999,mojiri2004}) and proposes a one-parameter robust
modification, which targets unbiased harmonic signals with a
slowly varying amplitude and band-limited white noise.

The rest of the paper is structured as follows. The problem
statement for a noisy harmonic signal with slowly varying
amplitude and an unknown frequency of interest is given in Section
\ref{sec:2}. The existing globally convergent frequency estimator
(relevant to this work) is summarized in Section \ref{sec:3}. In
this regard, differences in the proposed estimator (and therefore
the contributions of the paper) are highlighted. The main results,
with corresponding analysis and proofs, are provided in Section
\ref{sec:4}. In Section \ref{sec:5}, various simulated and
experimental harmonic signals are shown as confirming the
highlighted estimator properties. Brief conclusions are drawn in
Section \ref{sec:6}.

%%%%%%%%%%%%%%%%%%%%%%%%%%%%%%%%%%%%%%%%%%%%%%%%%%%%%%%%%%%%%%%%%%%%%%%%%%%%%%%%
\section{Problem statement}
\label{sec:2}

We consider a classical estimation problem, which is of importance
for system identification and adaptive control, where a signal of
the form
\begin{equation}
\sigma(t) = k(t) \sin(\omega_0 t) + \eta(t), \label{eq:2:1}
\end{equation}
is the single measured oscillating quantity. The harmonic signal
has a slowly varying\footnotemark{\footnotetext{For the rest of
the paper, we will assume a sufficiently slow amplitude variation,
i.e., in terms of a low $|\dot{k}|$, compared to the basic
frequency $\omega_0$ of the sinusoidal signal. Using the averaging
theory of the periodic signals, we will assume that the system
\eqref{eq:2:1} contains timescale separation, which means a faster
oscillation with the angular frequency $\omega_0$ versus a slower
drift of $k(t)$.}} amplitude within a certain range $\underline{k}
\leq k \leq \overline{k}$, and an unknown angular frequency
$\omega_0 > 0$, which we are mainly interested in. The measured
$\sigma(t)$ is affected by the noise $\eta(t)$, which is a
zero-mean ergodic process uniformly distributed over the whole
frequency range $\omega$. In other words, $\eta(t)$ can be seen as
power-limited (and therefore band-limited) white noise, seen from
a signal-processing viewpoint. We will assume a constant power
spectral density (PSD) of the noise, i.e., $\mathrm{PSD} \{
\eta(\omega) \} \equiv p = \mathrm{const}$, and a reasonable (in
terms of the signal-to-noise ratio) finite variance $\mathrm{Var}
\{ \eta(t) \} = \tau^2$, this without going into further details
about the spectral properties of $\eta(j \omega)$.

Although biased sinusoids with unknown frequency are often
considered (e.g., \cite{marino2003,aranovskiy2010}), this work
focuses only on unbiased sinusoidal signals \eqref{eq:2:1}, while
keeping in mind that some constant bias can be removed by
high-pass or other dedicated filtering approaches. Rather, we
emphasize that $k(t)$ can be slowly varying. For instance, if one
allows for $k(t) \rightarrow 0$ with the progressing time, then
\eqref{eq:2:1} will represent a damped oscillation response
$\sigma(t)$, where only a finite number of the periods is
available for estimating $\omega_0$.

%%%%%%%%%%%%%%%%%%%%%%%%%%%%%%%%%%%%%%%%%%%%%%%%%%%%%%%%%%%%%%%%%%%%%%%%%%%%%%%%
\section{Globally convergent adaptive notch filter}
\label{sec:3}

The globally convergent adaptive filter (also denoted as an
adaptive notch filter (ANF) due to the structural properties of a
second-order notch filter) was provided and analyzed in
\cite{hsu1999}, based on the original work \cite{regalia1991}. A
continuous-time version of the ANF \cite{regalia1991} can also be
found in \cite{bodson1997}. The ANF considers a sinusoidal signal
\eqref{eq:2:1} but without explicit amplitude variation $k(t)$ and
noise $\eta(t)$, and estimates the unknown frequency $\omega_0$
using the following structure \cite{hsu1999}:
\begin{eqnarray}
\label{eq:3:1}
  \ddot{x} + 2\zeta \theta \dot{x} + \theta^2 x &=& \theta^2 \sigma, \\
  \dot{\theta} &=& -\gamma x \bigl( \theta^2 \sigma  -  2\zeta \theta \dot{x}
  \bigr).
\label{eq:3:2}
\end{eqnarray}
Subsequently, another scaling of the forcing signal in
\eqref{eq:3:1} and, correspondingly, error signal in
\eqref{eq:3:2} was proposed and analyzed in \cite{mojiri2004},
while it was claimed that the adaptation scheme and necessary
stability condition become independent of the damping ratio
$\zeta$. Both aforementioned formulations of the ANF have two real
positive design parameters: $\gamma$, which determines the
'adaptation speed', and $\zeta$, which determines the 'depth of
the notch' and hence noise sensitivity, according to
\cite{mojiri2004}. In both approaches \cite{hsu1999,mojiri2004},
stability analysis and proof of convergence rely on the concept of
the uniqueness of a periodic orbit $\bigl[ \bar{x}, \bar{\dot{x}},
\bar{\theta} \bigr](t)$, where $(\bar{\cdot})$ denotes equilibria
(not necessarily constant). Towards this unique orbit, with
$\bar{\theta} = \omega_0$, the adaptive system \eqref{eq:3:1},
\eqref{eq:3:2} is then shown to converge globally. Since its
appearance, the ANF approach has become popular, and performance,
design equations and applications, as well as signal/noise ratios
and initialization, have been addressed in several works (see,
e.g., \cite{clarke2001}).

The estimator introduced below keeps the same motivating ANF
principle while (i) adapting the scaling factor, (ii) using the
sign instead of the $x$-state in \eqref{eq:3:2}, and (iii)
canceling the damping ratio design parameter, which is shown to be
unnecessary. The proposed convergence analysis is based on a simpler
consideration of steady states in the frequency domain, through
keeping $\theta$ as a frozen parameter, which avoids the challenges of
demonstrating convergence of \eqref{eq:3:1}, \eqref{eq:3:2} into a
unique periodic orbit (cf. \cite{hsu1999,mojiri2004}). At the same
time, we can demonstrate explicitly an exponential convergence rate
and can take into account (explicitly) the impact of measurement noise
and (implicitly) insensitivity to slow amplitude variations.

%%%%%%%%%%%%%%%%%%%%%%%%%%%%%%%%%%%%%%%%%%%%%%%%%%%%%%%%%%%%%%%%%%%%%%%%%%%%%%%%
\section{Main results}
\label{sec:4}

The proposed robust global frequency estimator is provided by the
auxiliary second-order system
\begin{eqnarray}
\label{eq:4:1}
\left[%
\begin{array}{c}
  \dot{x}_1 \\
  \dot{x}_2 \\
\end{array}%
\right] & = & \left[%
\begin{array}{cc}
  0 & 1 \\
  -\theta^2 & -2 \theta \\
\end{array}%
\right] \left[%
\begin{array}{c}
  x_1 \\
  x_2 \\
\end{array}%
\right] + \left[%
\begin{array}{c}
  0 \\
  2 \theta \\
\end{array}%
\right] \sigma,
\\[1mm]
  y & = & \left[%
\begin{array}{cc}
  0 & 1 \\
\end{array}%
\right] \left[%
\begin{array}{c}
  x_1 \\
  x_2 \\
\end{array}%
\right], \label{eq:4:2}
\end{eqnarray}
and the adaptation law
\begin{equation}
\dot{\theta} = -\gamma \, \mathrm{sign}(x_1)(\sigma - y),
\label{eq:4:3}
\end{equation}
where $\gamma > 0$ appears as the single design parameter. Being
excited by $\sigma(t)$, the vector of dynamic states $[x_1,
x_2]^T$ performs steady-state oscillations at the angular
frequency $\theta = \omega_0$, once the input-output
synchronization brings the \emph{output error} to zero, i.e., $e =
\sigma - y \rightarrow 0$. For a slowly varying $k(t)$, the
$2 \theta$ input coupling factor used (on the right-hand side of
\eqref{eq:4:1}) endows the input-output ratio of \eqref{eq:4:1},
\eqref{eq:4:2} to be independent of $k$ in a steady state, thus
making a frequency estimation insensitive to slow amplitude
variations. It is also worth noting that, unlike similar
adaptation mechanisms \cite{hsu1999}, \cite{mojiri2004}, no
further scaling factors are assigned to the forcing signal
$\sigma(t)$. A clear advantage of this purposeful simplification
will be shown later when analyzing the convergence rate. While an
ANF contains an additional damping parameter, which determines the
'depth of the notch' and, according to \cite{hsu1999,mojiri2004},
its noise sensitivity, we deliberately assume a critically damped
dynamic system \eqref{eq:4:1}, \eqref{eq:4:2} (compared with
\eqref{eq:3:1} for $\zeta=1$). This simplification not only
removes the second design parameter but also eliminates transient
oscillations of $\theta(t)$ in the course of the frequency adaptation.
This comes as no surprise since the linear $[x_1, x_2]^T$
sub-dynamics with $\zeta=1$ have no conjugate-complex poles and
therefore no transient oscillations (thus also not in $e(t)$).
Later, we will demonstrate a performance degradation when a
damping factor $0 < \zeta < 1$ is included, as a parameter, in
\eqref{eq:4:1}, \eqref{eq:4:2}.

Denoting the system matrix and input and output coupling vectors
in \eqref{eq:4:1}, \eqref{eq:4:2} by $A$, $B$ and $C$,
correspondingly, the input-to-output transfer function
\begin{equation}
G(j\omega) = \frac{y(j\omega)}{\sigma(j\omega)} = C^T (j\omega
I-A)^{-1}B, \label{eq:4:4}
\end{equation}
can be written for the frequency domain, when $\theta$ is
considered to be a frozen parameter. Obviously, $I$ is the $2
\times 2$ identity matrix, and $\omega$ is the angular frequency
variable. As long as the output error is not zero, it is excited
as
\begin{equation}
e(j\omega) = (1-G) G^{-1} y(j\omega), \label{eq:4:5}
\end{equation}
by the forced dynamics \eqref{eq:4:1}, \eqref{eq:4:2}, so that
$\bigl | e(j\omega)   \bigr | > 0$ for all $\omega$ excepts
$\omega = \omega_0$. This motivates the adaptation law
\eqref{eq:4:3}, which ensures $\dot{\theta} \neq 0$ always excepts
at $\theta=\omega_0$. Denoting the above transfer function, i.e.,
from the filter output to the error, by
$$
E(j\omega) = \frac{e(j\omega)}{y(j\omega)} =
\bigl(1-G(j\omega)\bigr) G(j\omega)^{-1},
$$
we can take a closer look at the amplitude and phase response of
$E(\cdot)$, illustrated in Fig. \ref{fig:1} for $\theta =
10$ rad/sec.
\begin{figure}[!h]
\centering
\includegraphics[width=0.98\columnwidth]{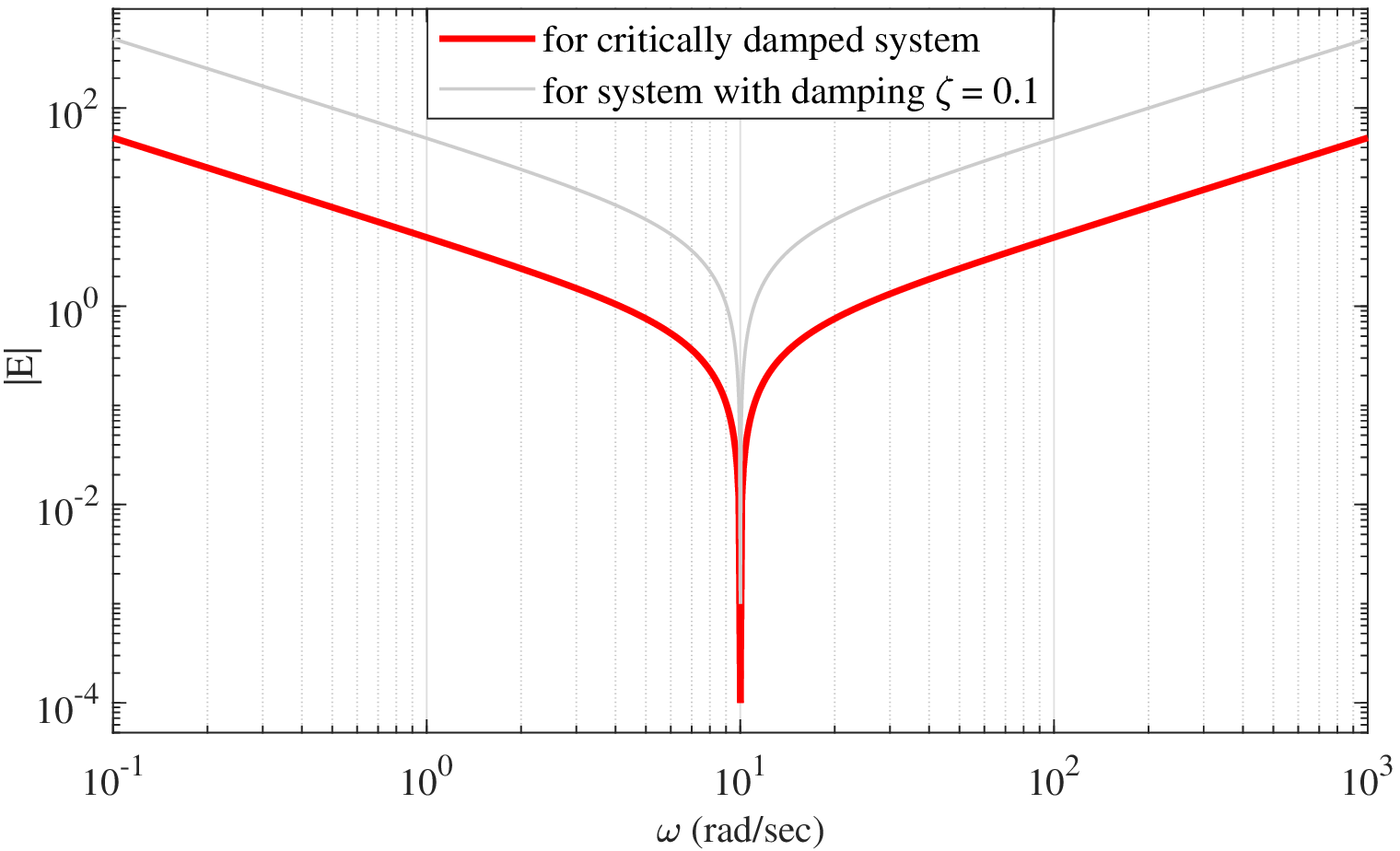}
\includegraphics[width=0.98\columnwidth]{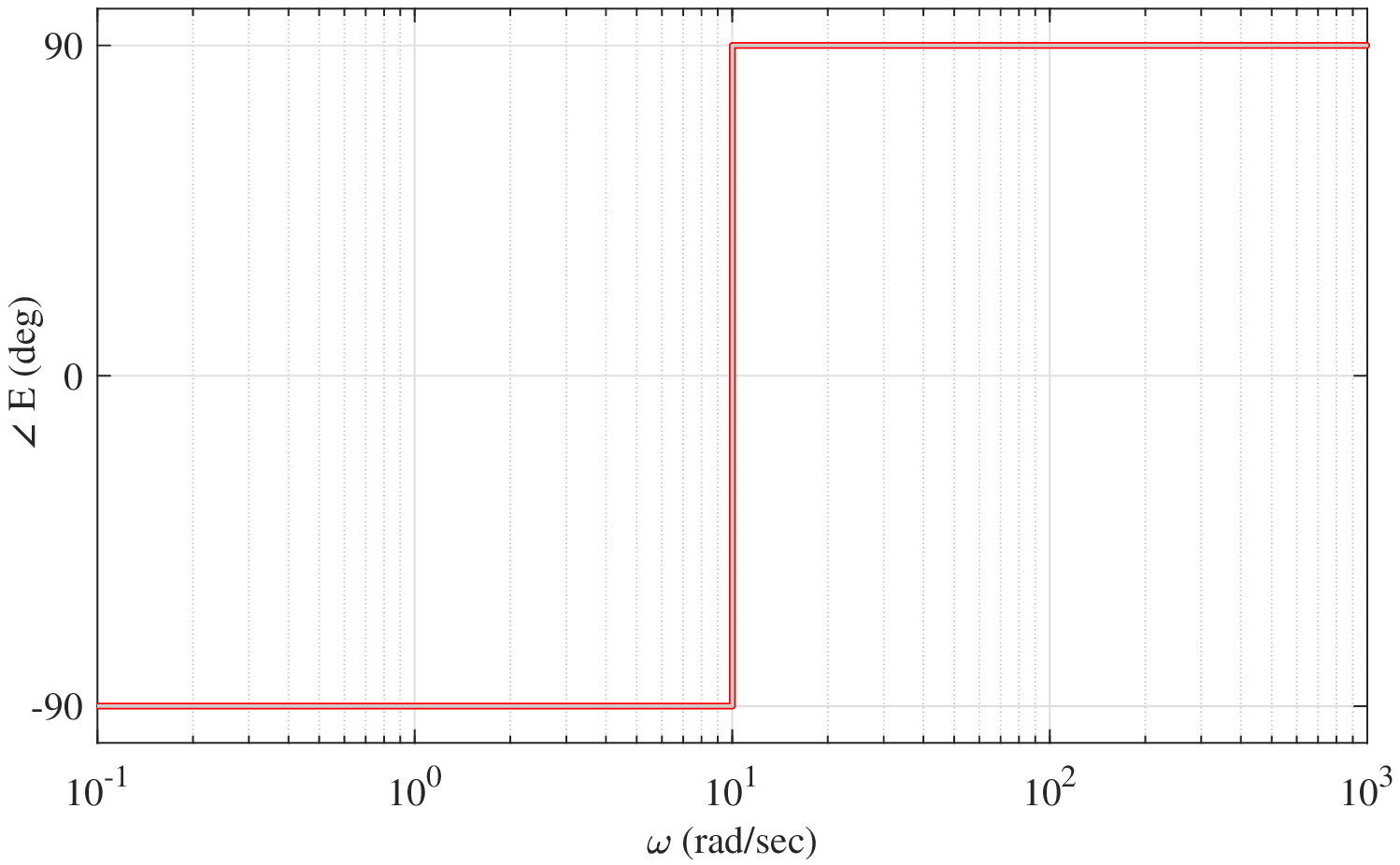}
\caption{Amplitude and phase response of the transfer function
$E$, for exemplary $\theta = 10$ rad/sec; the transfer function
with an additional damping $\zeta = 0.1$ (cf. with (1)) is
included (gray line) for the sake of comparison.} \label{fig:1}
\end{figure}
It becomes evident that the \emph{error transfer function}
amplitude $|E|$ has global minima at $\omega = \theta$, so that
$|e(j\omega)| \rightarrow 0$ when $\theta(t) \rightarrow
\omega_0$, and this is independent of the estimator initialization
$\theta(0)$. Indeed, without loss of generality, we can assume an
arbitrary $0 < \theta(0) \neq \omega_0$ so that $\bigl | e(j
\omega) \bigr |_{\omega=\omega_0} = a$, where $a > 0$ is some
positive magnitude determined by the oscillating output state and
error transfer function. We also recall that the oscillating
output in a steady state will then be given by $y = G(j\omega)
\sigma$, since the system \eqref{eq:4:1}, \eqref{eq:4:2} is
asymptotically stable and, moreover, critically damped. This
basically leads to the harmonic behavior of $y(t)$ and $e(t)$,
provided $\omega_0 = \mathrm{const}.$

Further, it becomes evident (cf. phase response in Fig.
\ref{fig:1}) that $\angle e(j \omega)$ always lags behind the
phase $\angle y(j \omega)$ for $\omega < \omega_0$, which is due to
$\angle E(j \omega) \rightarrow -\pi/2$. Correspondingly, $\angle
e(j \omega)$ always leads before the phase $\angle y(j \omega)$
for $\omega > \omega_0$, which is due to $\angle E(j \omega)
\rightarrow \pi/2$. It is only at $\omega = \omega_0$ where both signals
$y(j \omega)$ and $e(j \omega)$ are in the phase and $e
\rightarrow 0$. The $\pm \pi/2$ phase response of $\angle E(j
\omega)$ allows for providing an ever-increasing or decreasing
$\theta(t)$ on the left-hand or, respectively, right-hand side
of $\omega_0$. With this in mind, we are now in a position to
formally prove global convergence of the adaptation law
\eqref{eq:4:3}.

\begin{thm} %[\emph{Name of Theorem}]
\label{thm:1} The frequency estimator
\eqref{eq:4:1}-\eqref{eq:4:3} is global for \eqref{eq:2:1} and
converges asymptotically as $\theta(t) \rightarrow \omega_0$ for
$t \rightarrow \infty$, regardless of the $\theta(0)
> 0$ initialization, provided the small adaptation gains $\gamma > 0$
and slowly varying amplitudes $k(t)$. The frequency-estimation error
$\varepsilon (t) = \omega_0 - \theta(t)$ converges uniformly and
exponentially in terms of
\begin{equation}
\bigl | \varepsilon(t_2) \bigr | < \alpha \bigl | \varepsilon(t_1)
\bigr | \exp \bigl( - \beta (t_2-t_1) \bigr) \label{eq:4:01}
\end{equation}
for some $\alpha > 0$ and $\forall \; t_2>t_1$. The exponential
rate of convergence is independent of $\eta(t)$ noise, as follows:
\begin{equation}
\beta = 0.5 \, \gamma k \, \omega_0^{-1} + \delta, \label{eq:4:02}
\end{equation}
where $\delta$ is a small positive constant independent of
$\gamma$, $k$, $\omega_0$.
\end{thm}

\begin{proof} %[\emph{Name of Theorem}]
Let $0 < \theta(0) < \omega_0$ be an arbitrary initialization of
the estimator \eqref{eq:4:1}-\eqref{eq:4:3}. Note that for
$\theta(0) > \omega_0$, the proof is fully identical due to the
phase symmetry of $\angle E\bigl(j \theta \bigr) \rightarrow \pm
\pi/2$ for all $\theta \neq \omega_0$, and as a consequence,
$\mathrm{sign} \bigl( \dot{\theta} \bigr) = \pm 1$. A harmonic
excitation \eqref{eq:2:1} leads to an output harmonic
\begin{equation}
y(t) = b \sin (\omega_0 t + c), \label{eq:4:6}
\end{equation}
where $b = k \bigl | G(j\omega_0) \bigr|$, while the phase shift
$c = \angle \bigl( G(j \omega_0) \bigr)$ is of minor relevance
here. The internal dynamic state of the estimator then becomes
\begin{equation}
x_1 (t) = - \frac{b}{\omega_0} \cos(\omega_0 t+c) = -
\frac{b}{\omega_0} \sin(\omega_0 t + c + \pi/2), \label{eq:4:7}
\end{equation}
and the output error becomes
\begin{equation}
e(t) = a \sin (\omega_0 t + c + \pi/2), \label{eq:4:8}
\end{equation}
where $a = b | E(j \omega_0)| > 0$ for all $\theta < \omega_0$.
Substituting \eqref{eq:4:7} and \eqref{eq:4:8} into
\eqref{eq:4:3}, and writing out $a$ and $b$, results in
\begin{equation}
\dot{\theta} = \gamma k  \, \bigl | G(j \omega_0) \bigr | \, \Bigl
| \frac{1-G(j \omega_0) }{G(j \omega_0) } \Bigr | \, \bigl | \sin
(\omega_0 t + c + \pi/2) \bigr |. \label{eq:4:9}
\end{equation}
It is clear that for all $\gamma, k > 0$ the
$\mathrm{sign}( \dot{\theta}) = +1$ as long as $\theta(t) <
\omega_0$. This implies global uniform convergence and
completes the first part of the proof.

Evaluating the first and second modulus $|\cdot|$-terms in
\eqref{eq:4:9}
\begin{equation}
\bigl | G(j \omega_0) \bigr | \, \Bigl | \frac{1-G(j \omega_0)
}{G(j \omega_0) } \Bigr | = \frac{ \bigl |  \theta^2 - \omega_0^2
\bigr | }{ \theta^2 + \omega_0^2 } \equiv \Omega(\theta),
\label{eq:4:10}
\end{equation}
one can show that the $\theta$-dependent magnitude $\Omega$ always decreases monotonically and $1 \geq \Omega(\theta) \geq 0$
on the interval $ \theta \in [0, \, \omega_0]$. Inspecting the
$\Omega(\theta)$ function, with the $\omega_0$-normalized argument,
as depicted in Fig. \ref{fig:2},
\begin{figure}[!h]
\centering
\includegraphics[width=0.98\columnwidth]{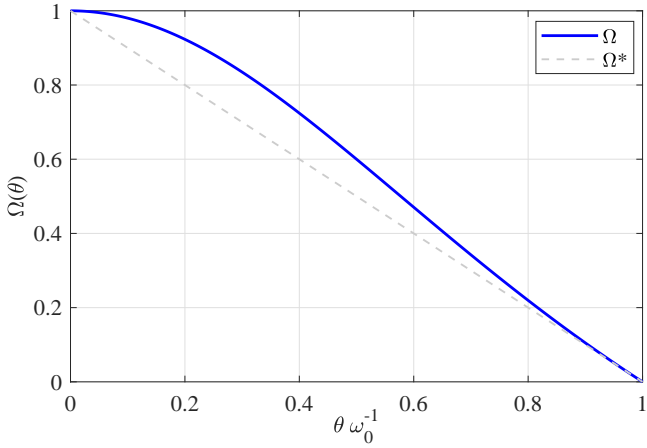}
\caption{$\Omega(\theta)$ function of the $\omega_0$-normalized
argument.} \label{fig:2}
\end{figure}
one can linearly approximate the magnitude of the decrease by
\begin{equation}
\Omega^*(\theta) = - \theta \omega_0^{-1} + 1. \label{eq:4:11}
\end{equation}
Using \eqref{eq:4:11} and the fact of an always positive third
modulus $|\cdot|$-term in \eqref{eq:4:9}, which has a mean value
equal to $1/2$, one can write the first-order approximation
\begin{equation}
\dot{\theta}^* = 0.5 \, \gamma k  \bigl ( -\theta^* \omega_0^{-1}
+ 1 \bigr) \label{eq:4:12}
\end{equation}
of the estimator dynamics. Note that the $\angle E\bigl(j \theta
\bigr) \rightarrow \pm \pi/2$ phase, which determines the sign of
$\dot{\theta}$ (cf. \eqref{eq:4:9}), allows us to write
\eqref{eq:4:12}, regardless of whether $\theta^*(0) < \omega_0$ or
$\theta^*(0) > \omega_0$. Because of \eqref{eq:4:11} is an
under-approximation of the $(\theta \omega_0^{-1})$-dependent
gaining factor of the adaptation rate (cf. Fig. \ref{fig:2}), the
following can be concluded from the eq. \eqref{eq:4:12}. The
asymptotic convergence of $\theta(t)$ to $\omega_0$ has an
exponential rate of $0.5 \gamma k \omega_0^{-1} + \delta$, which
is not slower than that of the dynamics
\begin{equation}
\dot{\theta}(t) + 0.5 \gamma k \omega_0^{-1} \theta(t) = 0.5
\gamma k \omega_0^{-1} \cdot \omega_0, \label{eq:4:13}
\end{equation}
(cf. \eqref{eq:4:12} and \eqref{eq:4:13}). This completes the second
part of the proof.
\end{proof}

\begin{rmk}
\label{rmk:1}

Note that the asymptotic convergence of $\theta(t)$ can be
guaranteed firstly when $y(t)$ is in a steady state, i.e., after
the transients of \eqref{eq:4:1}, \eqref{eq:4:2} dynamics excited
by \eqref{eq:2:1}. The transient response results in a temporary
bias of the output harmonic $y(t)$, depending on the initial phase
of the excitation signal $\sigma(t)$. Consequently, the
$\theta(t)$ trajectory can drift oscillatory in the opposite
direction, away from $\omega_0$, until $y(t)$ is in a steady state
(cf. the first numerical example in Section \ref{sec:5:sub:1}).
This initial by-effect must be taken into account when assigning
$0 < \theta(0) < \omega_0$, while being irrelevant for $\theta(0)
> \omega_0$.
\end{rmk}

\begin{rmk}
\label{rmk:2}

When including the damping factor $0 < \zeta < 1$ as an additional
design parameter of the estimator (cf. \cite{hsu1999,mojiri2004}),
the dynamic system \eqref{eq:4:1}, \eqref{eq:4:2} needs to be
modified in respect of the $A$ and $B$ terms (cf. with
\eqref{eq:3:1}) and consequently becomes non-critically damped. By
implication, additional oscillating dynamics of $y(t)$ appear,
both during the transients and with a steady state of the dynamic
$\theta(t)$ trajectory. Notice that $\zeta$ has a minor influence
on the asymptotic convergence and its exponential rate (cf. Fig.
\ref{fig:1} for $\zeta = 0.1$). At the same time, it deteriorates
the smoothness of $\theta(t)$ during the transient phase and
produces residual steady-state oscillations of $\varepsilon(t)$
(cf. $\theta(t)$-trajectories, exemplified in Fig. \ref{fig:2x}
for $\zeta = \{0.1,0.5,1\}$).
\begin{figure}[!h]
\centering
\includegraphics[width=0.98\columnwidth]{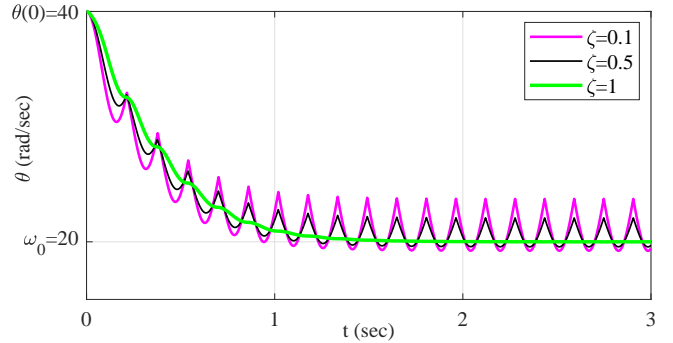}
\caption{Convergence trajectories of the estimator ($\gamma=100$,
$\omega_0=20$ rad/sec) with additional damping $\zeta =
\{0.1,0.5,1\}$.} \label{fig:2x}
\end{figure}
\end{rmk}

Once we have shown the asymptotic convergence of $\theta(t)$ to
$\omega_0$ and estimated the exponential convergence rate, it is
of further interest to analyze the non-vanishing residual
$\varepsilon(t)$, depending on the signal noise $\eta(t)$.

\begin{lem}
\label{lem:1}

The residual frequency-estimation error is a zero-mean ergodic
process, with
\begin{equation}
\mathrm{Var} \{ \varepsilon(t) \} < 4\tau^2 \omega_0 k^{-1},
\label{eq:4:14}
\end{equation}
for the signal \eqref{eq:2:1} with band-limited white noise,
which has variance $\mathrm{Var} \{ \eta(t) \} = \tau^2$.
\end{lem}

Note that the Lemma \ref{lem:1} claims the upper bound of the
second moment of $\varepsilon(t)$ for all times $t > t_c > 0$,
since $\varepsilon(t)$ is a random process driven by $\eta(t)$,
after $\theta(t)$ has converged to a neighborhood of $\omega_0$ at
some finite time $t_c$.

\begin{proof}
Denoting the harmonic part of the signal \eqref{eq:2:1} by
$\tilde{\sigma}(j\omega_0)$ and that of the output \eqref{eq:4:2}
by $\tilde{y}(j\omega_0)$, respectively, the output error in
the frequency domain can be written as
\begin{equation}
e(j\omega) = \tilde{\sigma}(j\omega_0) - \tilde{y}(j\omega_0) +
\eta(j\omega) - G(j\omega)\eta(j\omega). \label{eq:4:15}
\end{equation}
Provided the estimator has already converged to a neighborhood of
$\omega_0$, the harmonic part of the error can be set to zero, and
the residual error, due to the noise, is to be analyzed further.
Since no phase response can be considered for a stochastic noise
signal (only the magnitude), one can assume
\begin{equation}
|\hat{e}(j\omega)| = 2 |\eta(j\omega)| \label{eq:4:16}
\end{equation}
as a worst case (i.e., upper bound) since
$\|G(j\omega)\|_{\infty}=1$. Using the variance (for the noise
magnitude) and substituting it into \eqref{eq:4:3}, one can write
the noise-driven dynamics of the estimate as
\begin{equation}
|\dot{\hat{\theta}}| = \gamma \, 2 \tau^2,  \label{eq:4:17}
\end{equation}
for some neighborhood $\hat{\theta}$ of the true value $\omega_0$.
Note that the sign of $\dot{\hat{\theta}}$ is also a random
process driven by $\mathrm{sign}\bigl( x_1(t,\eta) \bigr) \,
\mathrm{sign}\bigl(\eta(t)\bigr)$ (cf. with \eqref{eq:4:3}). Now,
comparing \eqref{eq:4:9} and \eqref{eq:4:17}, one can see
that for the estimate dynamics not driven by noise, the
following inequality should hold
\begin{equation}
\mathrm{sign}\bigl(\dot{\theta}\bigr) = \mathrm{const} = \gamma \,
0.5 k \, \Omega(\theta) > \gamma \, 2 \tau^2. \label{eq:4:18}
\end{equation}
Using the linear approximation \eqref{eq:4:11}, and substituting
it into \eqref{eq:4:18}, yields
\begin{equation}
k \Bigl( - \frac{\theta}{\omega_0} + 1  \Bigr) > 4 \tau^2.
\label{eq:4:19}
\end{equation}
Solving \eqref{eq:4:19} with respect to $|\varepsilon|=|\omega_0 -
\theta(t)|$ results in
\begin{equation}
|\varepsilon| > 4 \tau^2 \omega_0 k^{-1}, \label{eq:4:20}
\end{equation}
which should be fulfilled so that the estimate dynamics
\eqref{eq:4:3} do not become driven by the signal noise. Turning
back to the random nature of the noise-driven residual estimation
error, one can state \eqref{eq:4:14} (cf. with \eqref{eq:4:20}),
which completes the proof.
\end{proof}

%%%%%%%%%%%%%%%%%%%%%%%%%%%%%%%%%%%%%%%%%%%%%%%%%%%%%%%%%%%%%%%%%%%%%%%%%%%%%%%%
\section{Numerical and experimental examples}
\label{sec:5}

\subsection{Simulated signals}
\label{sec:5:sub:1}

We firstly demonstrate convergence of the estimator
\eqref{eq:4:1}-\eqref{eq:4:3} for a purely sinusoidal signal
$\sigma(t) = \sin(\omega_0 t)$, i.e., with a constant unity
amplitude and without noise. Assuming two estimator
initializations $\theta(0) = \{10,100 \}$ rad/sec, i.e., one
higher and one lower than $\omega_0=50$ rad/sec, the $\theta(t)$
trajectories are shown in Fig. \ref{fig:3} for $\gamma = \{200,
100, 50\}$ adaptation gains.
\begin{figure}[!h]
\centering
\includegraphics[width=0.98\columnwidth]{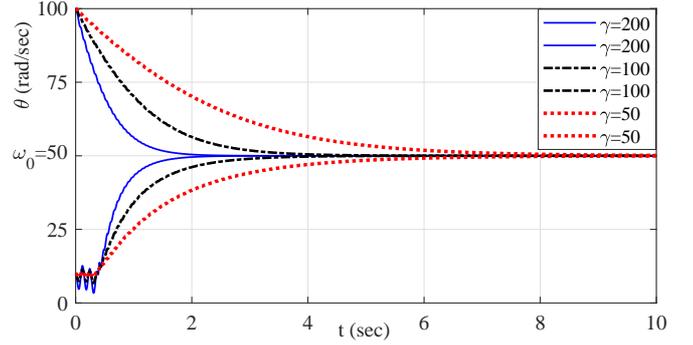}
\caption{Convergence of $\theta(t)$ with $\gamma = \{200, 100,
50\}$ to $\omega_0 = 50$ rad/sec for $\sigma(t) = \sin(\omega_0
t)$ signal with constant amplitude and without noise.}
\label{fig:3}
\end{figure}
The resulting exponential shape of convergence is in accord with
\eqref{eq:4:02}. Note that for $\theta(0) = 10$ rad/sec, the
initial $\theta(t)$ trajectory is oscillatory, progressing in
the opposite direction. This occurs due to a transient response of
\eqref{eq:4:1}, \eqref{eq:4:2} to the $\sigma$-excitation, which
takes about three oscillation periods for the given $\theta(0)$
and $\omega_0$ values (cf. Remark \ref{rmk:1}).

Next, we consider the signal \eqref{eq:2:1} with $k=1$ for
different $\omega_0 = \{10, 30, 50, 70 \}$ rad/sec. For all
angular frequencies, an additional band-limited white noise
$\eta(t)$ with $p=1e-7$ and $\tau^2 = 0.001$ is included, as
exemplified in Fig. \ref{fig:4} (a) for $\omega_0 = 10$
rad/sec. The adaptation gain is set to $\gamma = 100$.
\begin{figure}[!h]
\centering
\includegraphics[width=0.98\columnwidth]{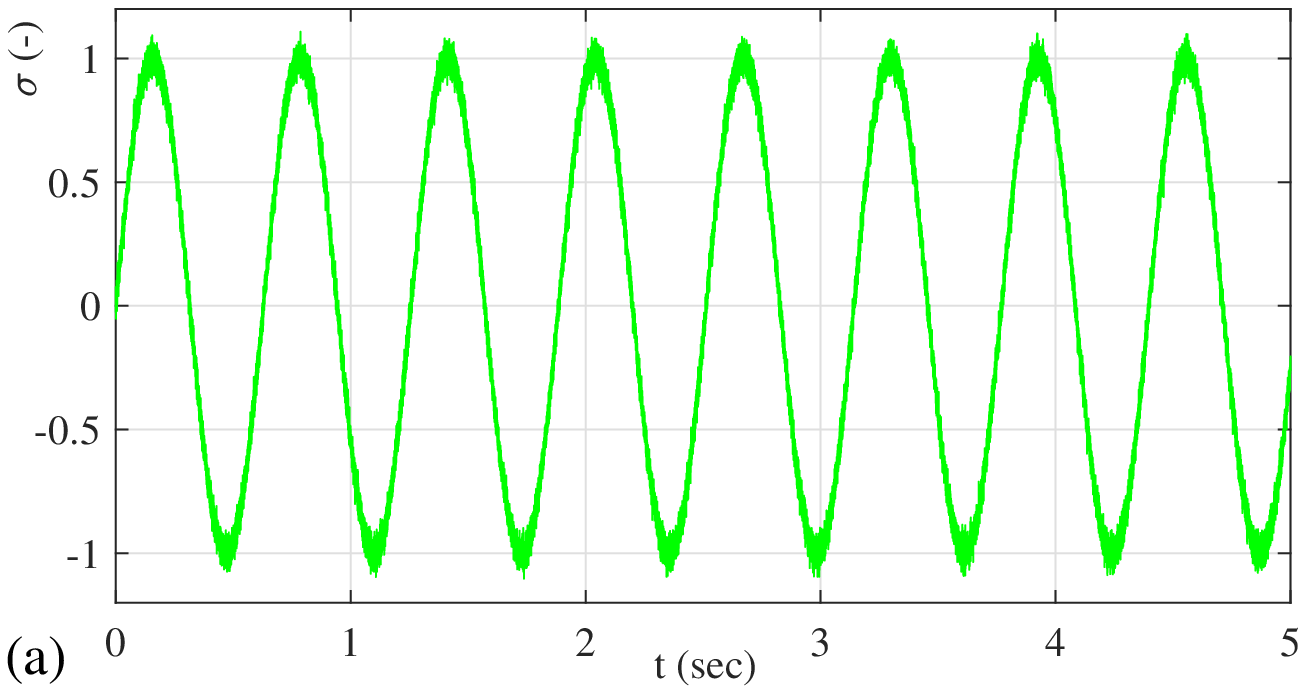}
\includegraphics[width=0.98\columnwidth]{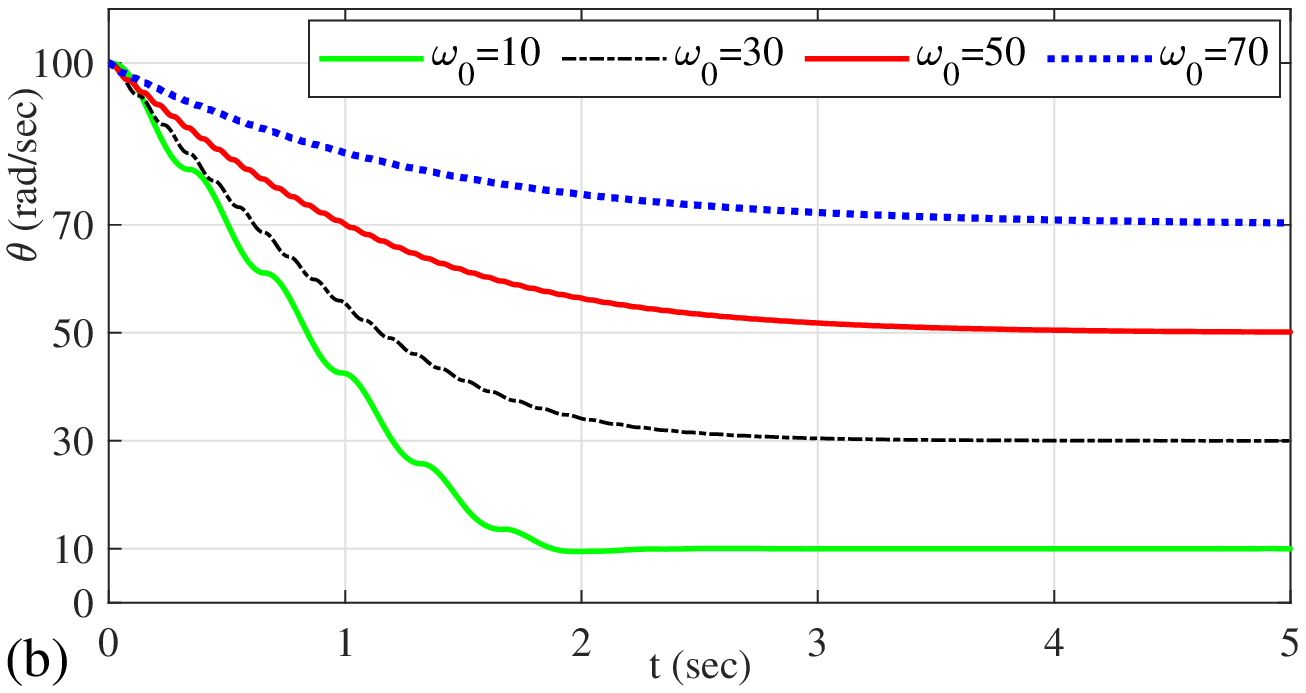}
\caption{(a) Simulated signal \eqref{eq:2:1} with $k=1$, $\omega_0
= 10$ rad/sec; (b) convergence of $\theta(t)$ with $\gamma = 100$
for $\omega_0 = \{10, 30, 50, 70 \}$ rad/sec.} \label{fig:4}
\end{figure}
The convergence of $\theta(t)$ depends inversely on $\omega_0$ and
is in accord with \eqref{eq:4:02}, as can be seen from Fig.
\ref{fig:4} (b). Note that the noise of $\eta(t)$ does not affect
the convergence rate but solely the steady-state fluctuations of
the residual estimation error $\varepsilon(t)$, in accord with the
Lemma \ref{lem:1}.

To demonstrate insensitivity of the frequency estimator to the
slow amplitude variations of the signal \eqref{eq:2:1}, we
consider $\omega_0 = 40$ rad/sec with $k(t)=5+5\sin(0.9t-\pi/2)$,
as depicted in Fig. \ref{fig:5} (a).
\begin{figure}[!h]
\centering
\includegraphics[width=0.98\columnwidth]{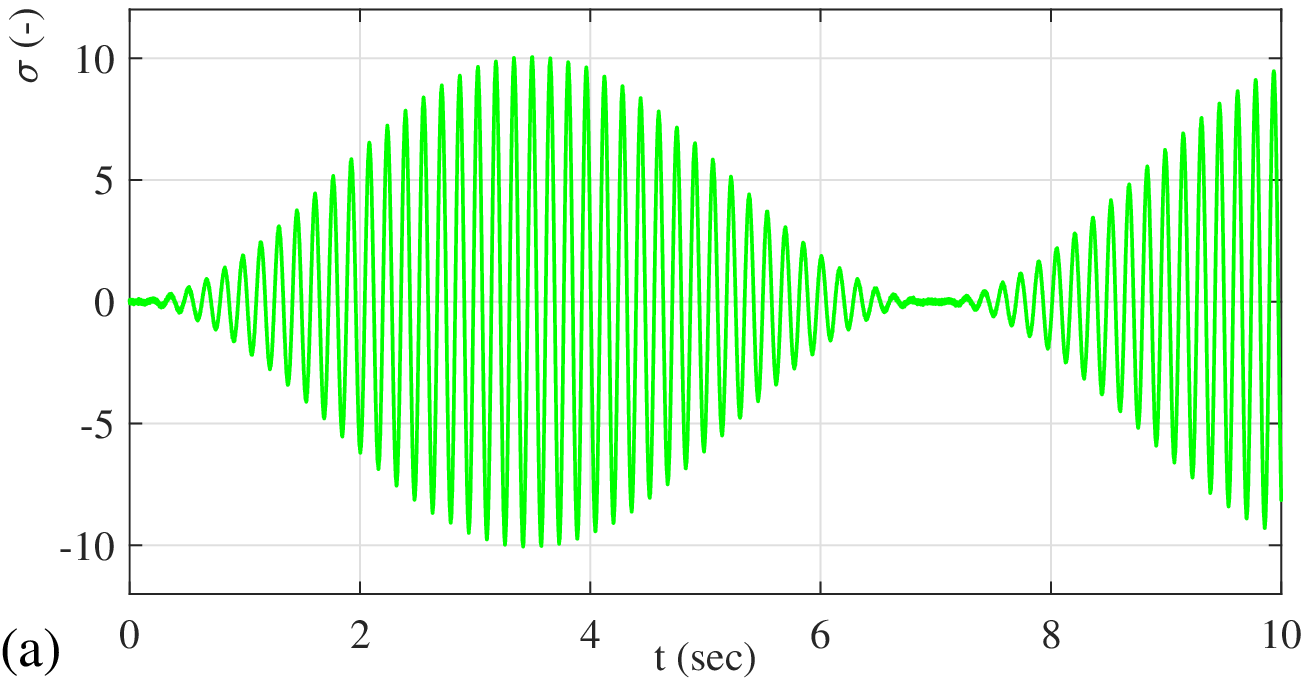}
\includegraphics[width=0.98\columnwidth]{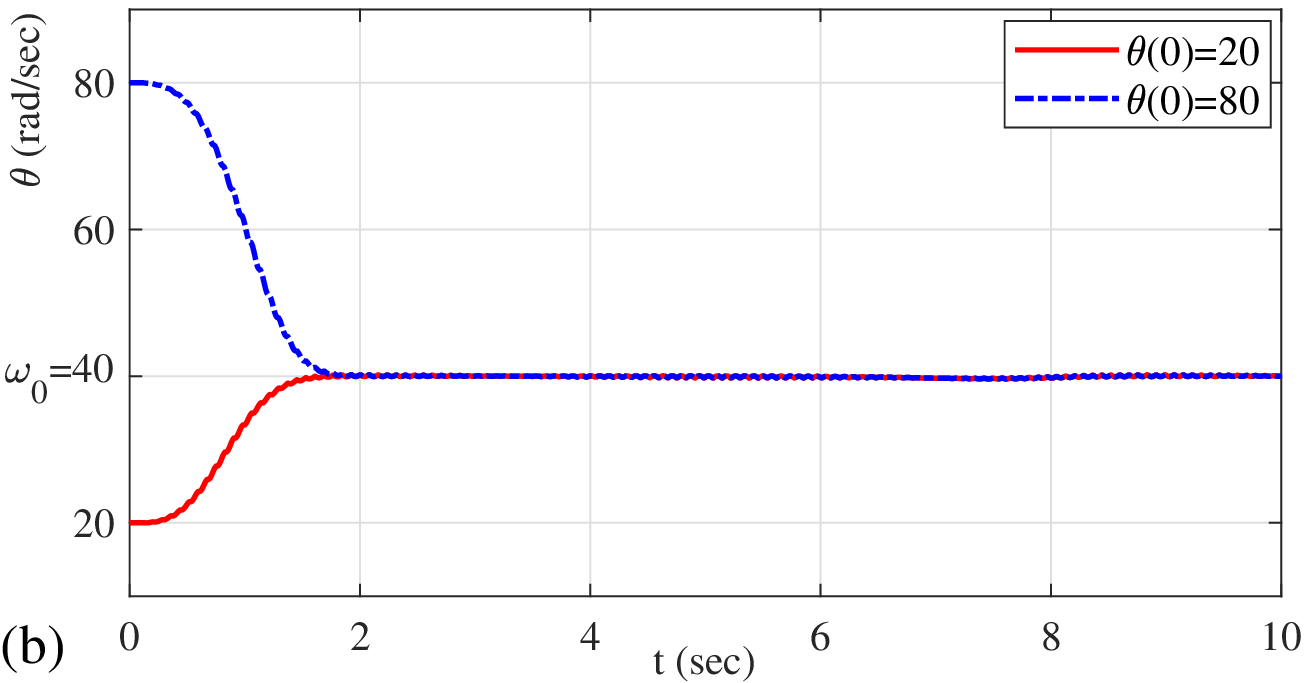}
\caption{(a) Simulated signal \eqref{eq:2:1} with
$k(t)=5+5\sin(0.9t-\pi/2)$ and $\omega_0 = 40$ rad/sec; (b)
convergence of $\theta(t)$ with $\gamma = 100$.} \label{fig:5}
\end{figure}
For the adaptation gain $\gamma = 100$ and two different estimator
initializations $\theta(0)=\{ 20, 80 \}$, the $\theta(t)$
convergence is shown in Fig. \ref{fig:5} (b). After a stable
transient of $\theta(t)$, whose shape is dynamically affected by
the $k(t)$ variations, the estimation error $\varepsilon(t)$ in a
steady state (for $t > 1.8$ sec) does not appear to be affected by persistent variations in the amplitude $k(t)$.

Finally, we are eager to see how the proposed robust estimator can
deal with the continuously varying frequencies $\omega_0(t)$. For
this purpose, the simulated signal \eqref{eq:2:1} with $k=1$ is
designed as a linear down-chirp $\omega_0(t) = \omega_0(0)-\mu t$,
with the frequency bounds $\omega_0(0)=20 \cdot 2\pi$ rad/sec and
$\omega_0(30)=1 \cdot 2\pi$ rad/sec, and the resulting $\mu =
3.98$, as depicted in Fig. \ref{fig:6} (a).
\begin{figure}[!h]
\centering
\includegraphics[width=0.98\columnwidth]{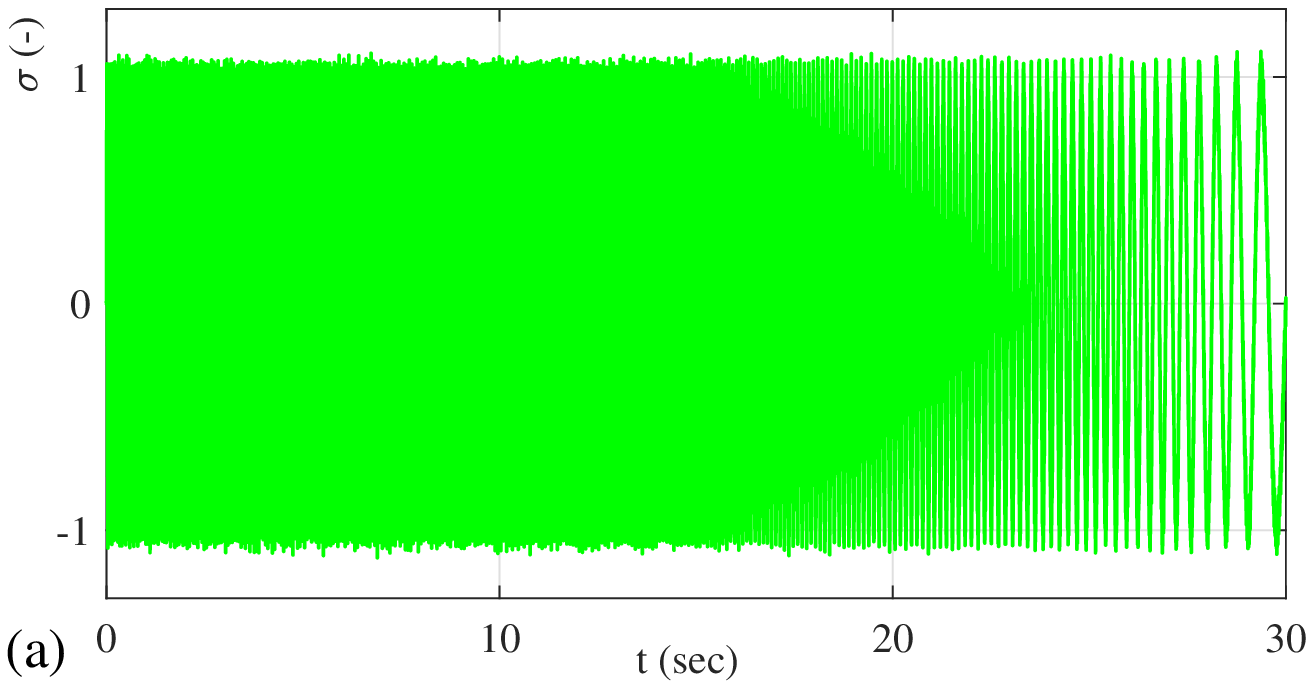}
\includegraphics[width=0.98\columnwidth]{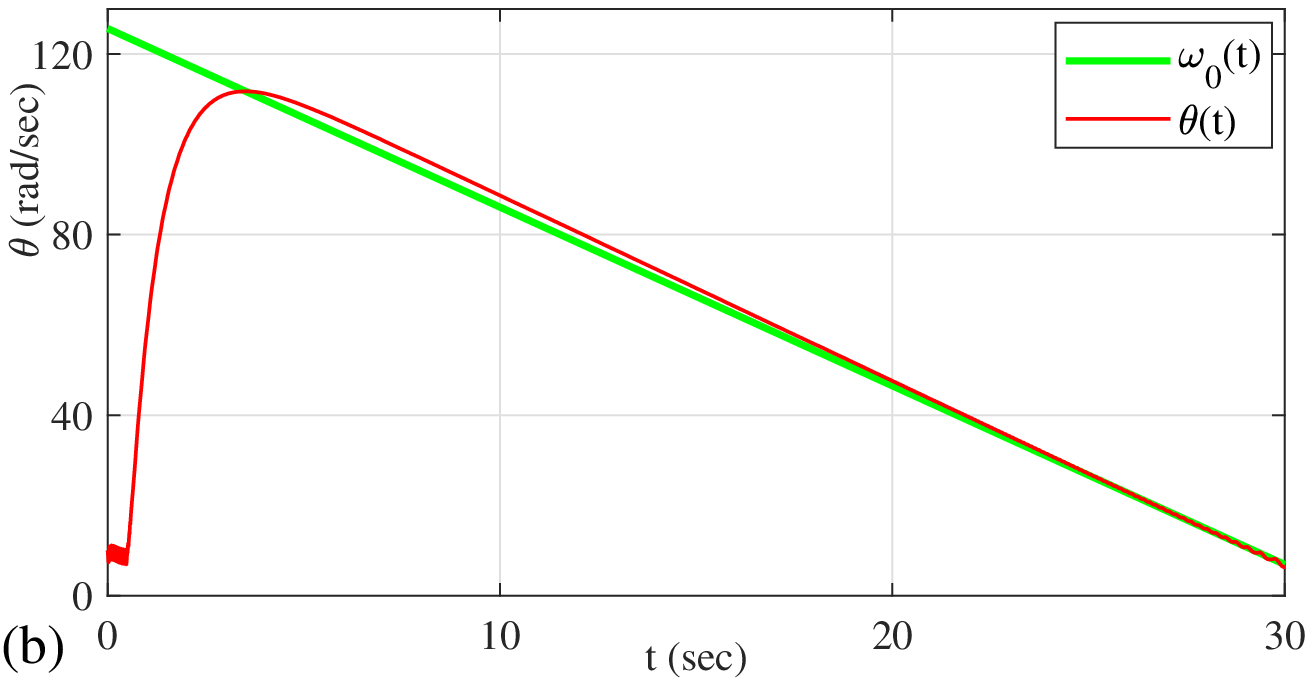}
\caption{(a) Simulated signal \eqref{eq:2:1} with $k=1$,
$\omega_0(t) = \omega_0(0)-\mu t$; (b) convergence of $\theta(t)$
with $\gamma = 200$, $\theta(0)=10$ rad/sec.} \label{fig:6}
\end{figure}
The $\theta(t)$ trajectory, for the assigned $\gamma = 200$, is
shown in Fig. \ref{fig:6} (b) versus the linearly changing
$\omega_0(t)$. It can be seen that after a certain time, the
$\theta(t)$ trajectory closely follows $\omega_0(t)$. The visible
residual estimation error $\varepsilon(t) \neq 0$ is clearly due
to the dynamically changing excitation frequency $\omega_0$; a
more detailed analysis of this effect is beyond the scope of this
work. Still, the estimator appears sufficiently robust to also follow
a continuously varying excitation frequency.

\subsection{Experimental case}
\label{sec:5:sub:2}

The proposed frequency estimation algorithm can equally be used
for mechanical system applications in which the oscillating
behavior (including vibrations) of structural parts and elements
with elasticities requires tracking of the frequency for various
purposes. Those include, e.g., controller tuning, condition and
fault monitoring, commissioning and identification, and others. An
experimental case provided below is realized in a laboratory
setting which is still representing a standard situation of an
unknown mechanical oscillation frequency in combination with a low
damping. Such application scenarios are commonly appearing in
two-inertia systems with either a load-depending varying natural
frequency, like in the flexible robotic joints (see e.g.
\cite{kim2019model,ruderman2020robotica}) and machine tools and
instruments with cantilevers and flexible frames (see e.g.
\cite{leonard2001free,beijen2018}). Or it is associated with
problems of varying excitation frequencies that propagate through
a flexible, correspondingly oscillating structure (see e.g.
\cite{baz2001active,helsen2013}).

For benchmarking with existing frequency estimators of the same
principle (cf. Sections \ref{sec:1}, \ref{sec:3}), the ANF
modified and proposed in \cite{mojiri2004} was also implemented,
in the same numerical setting as the proposed algorithm
\eqref{eq:4:1}-\eqref{eq:4:3}. Both frequency estimators are then
evaluated, as described below, on the same experimental data and
for the same initial and parametric conditions.

The case study with a simultaneous variation of the signal
amplitude $k(t)$ and angular frequency $\omega_0(t)$ is evaluated
experimentally suing the laboratory setup \cite{ruderman2020}
shown in Fig. \ref{fig:expsetup}. More details on the system
dynamics and modeling, that is however of lower relevance for the
recent work, an interested reader is referred to
\cite{ruderman2021}.
\begin{figure}[!h]
\centering (a) \includegraphics[width=0.45\columnwidth]{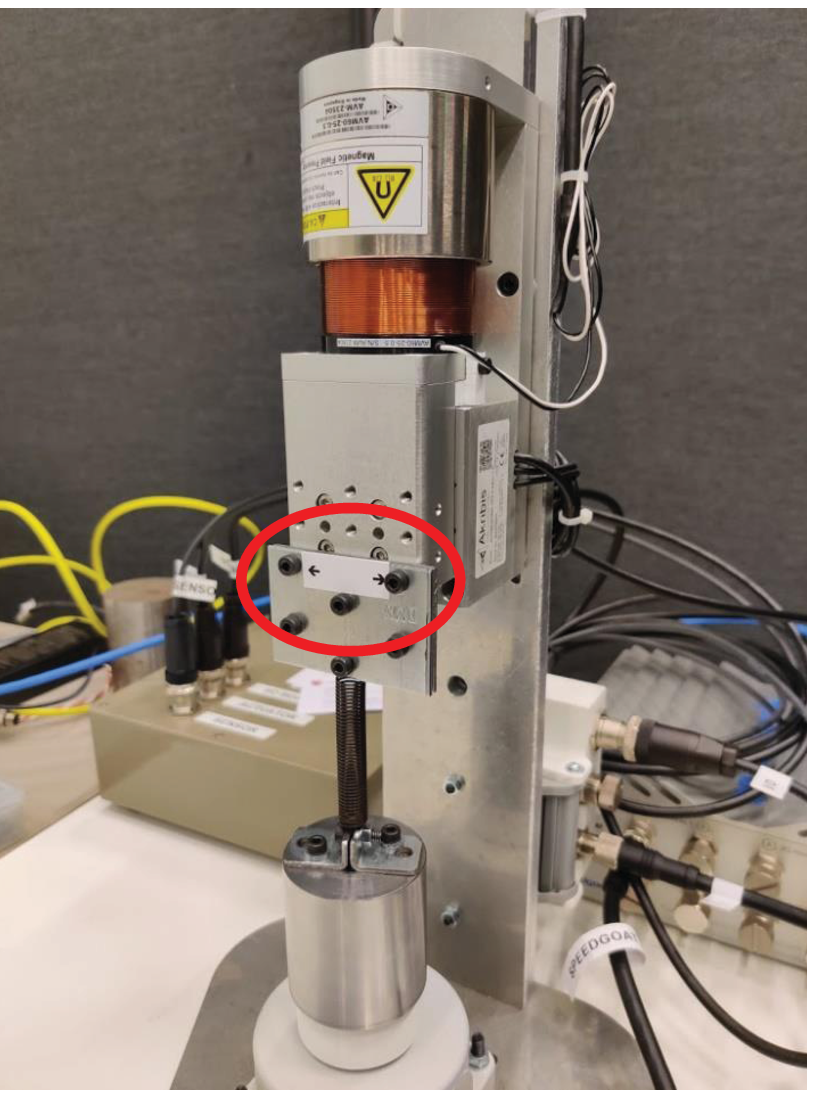}
\hspace{1mm} (b)
\includegraphics[width=0.32\columnwidth]{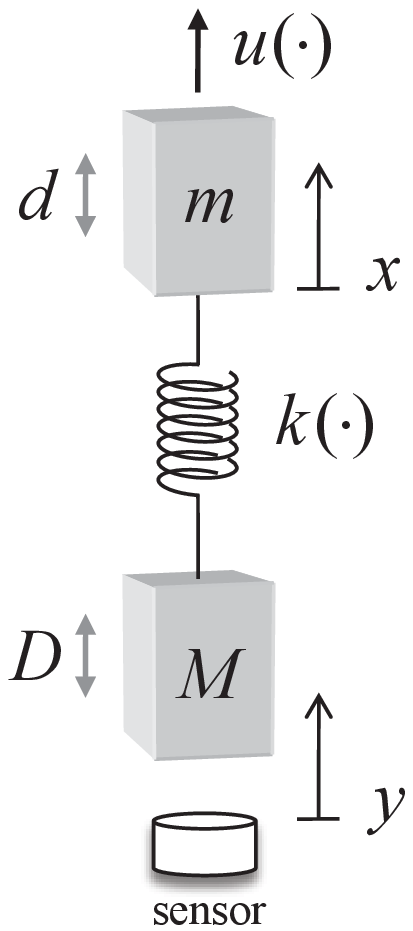}
\caption{Experimental setup: (a) laboratory view of the two-mass
oscillator, passive free-hanging mass ($M$) is placed on the
holder-disc when not in operation; (b) schematic representation of
two-mass oscillator, oscillating displacement $y \equiv \sigma(t)$
is measured contactlessly.}\label{fig:expsetup}
\end{figure}
The oscillating displacement of a free hanging load, attached
through a nearly linear spring ($K$), is measured contactlessly by
means of an inductive distance sensor, which has $\pm 12$ $\mu$m
repeatability. The sampling rate of real-time measurements, analog
to digital conversion with 16-bit quantization, is 2 kHz. Being
subject to both sensing ($\eta(t)$) and process disturbances
(non-modeled $d$ and $D$), the measured response $\sigma(t)$
constitutes an oscillating and noisy (cf. zoom-in in Fig.
\ref{fig:expinput}) time series. The double-mass experimental
setup, with the first 'active' mass ($m$) of the voice-coil-motor
actuator and second 'passive' mass ($M$) of a free hanging load,
allows for testing of the eigendynamics response and excited
(i.e., input-driven) response, both of a low-damped oscillating
nature.

The measured signal (Fig. \ref{fig:expinput}) represents an
exemplary response to an actuated chirp excitation, which yields a
linearly increasing angular frequency $\omega_0(t) = \mu t$, with
some initial value and $\mu>0$. The measured displacement
$\sigma(t)$ is freed (by data postprocessing) from a steady-state
offset, thus approaching a single harmonic, in accord with
\eqref{eq:2:1}. Note that apart from the noise, the measured
$\sigma(t)$ is still slightly affected by an asymmetry around
zero, therefore disclosing a certain additional non-constant bias.
\begin{figure}[!h]
\centering
\includegraphics[width=0.98\columnwidth]{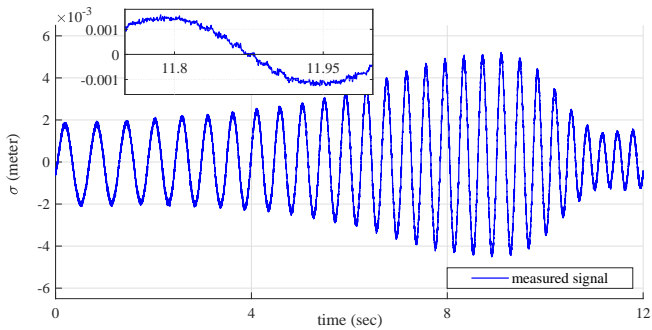}
\caption{Experimentally measured signal $\sigma(t)$ with varying
amplitude $k(t)$ and linearly increasing angular frequency
$\omega_0(t)$.} \label{fig:expinput}
\end{figure}

The evaluation setting, in terms of the parameters and initial
conditions, is summarized in Table \ref{tab:1} for both estimators
under benchmark.
\begin{table}[!h]
  \renewcommand{\arraystretch}{1.5}
  \caption{Estimators' evaluation setting}
  %\small
  \footnotesize
  \label{tab:1}
  \begin{center}
  \begin{tabular} {|p{2cm}|p{2.7cm}|p{2.7cm}|}
  \hline \hline
  Set parameter &   ANF according to \cite{mojiri2004}   &  Proposed estimator       \\
  \hline \hline
  $\theta(0)$ (rad/sec)    &   20    &  20     \\
  \hline
  $\gamma$                 &  $4e4$    &  $2e4$    \\
  \hline
  $\zeta$                 &  $\{0.7, \, 1,\, 1.3 \}$    &  $\{0.7,\, 1,\, 1.3 \}$    \\
  \hline \hline
  \end{tabular}
  \end{center}
  \normalsize
\end{table}
Note that for the proposed estimator, the assigned $\gamma$-gain
is twice smaller than $\gamma$ for the ANF \cite{mojiri2004},
since for the latter it is already integrating the factor
$2\zeta$, cf. \cite[eqs.~(5),(6)]{mojiri2004}. Further, for the
sake of completeness and a fair comparison, the damping ratio
$\zeta$ is additionally included into the proposed estimator.
Recall that, otherwise, $\zeta=1$ is set as default and is not
appearing as a design parameter, according to
\eqref{eq:4:1}-\eqref{eq:4:3}.

The online estimate of the angular frequency $\theta(t)$ is shown
in Fig. \ref{fig:expresults} versus the linear progress of the
chirp-driven true $\omega_0(t)$ value. The estimated $\theta(t)$
values are plotted over each other for the ANF \cite{mojiri2004}
and the proposed estimator, for $\zeta=0.7$ in (a), for $\zeta=1$
in (b), and for $\zeta=1.3$ in (c), correspondingly.
\begin{figure}[!h]
\centering
\includegraphics[width=0.98\columnwidth]{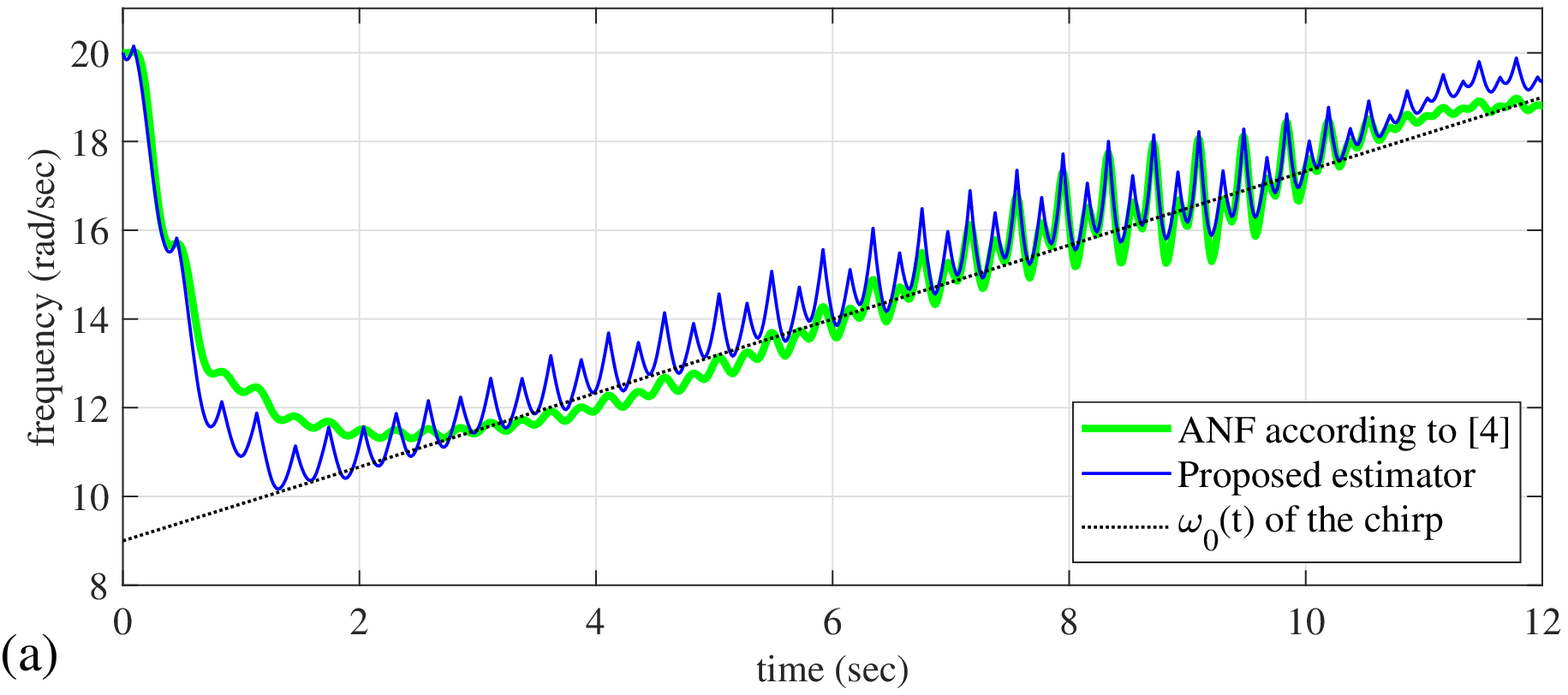}
\includegraphics[width=0.98\columnwidth]{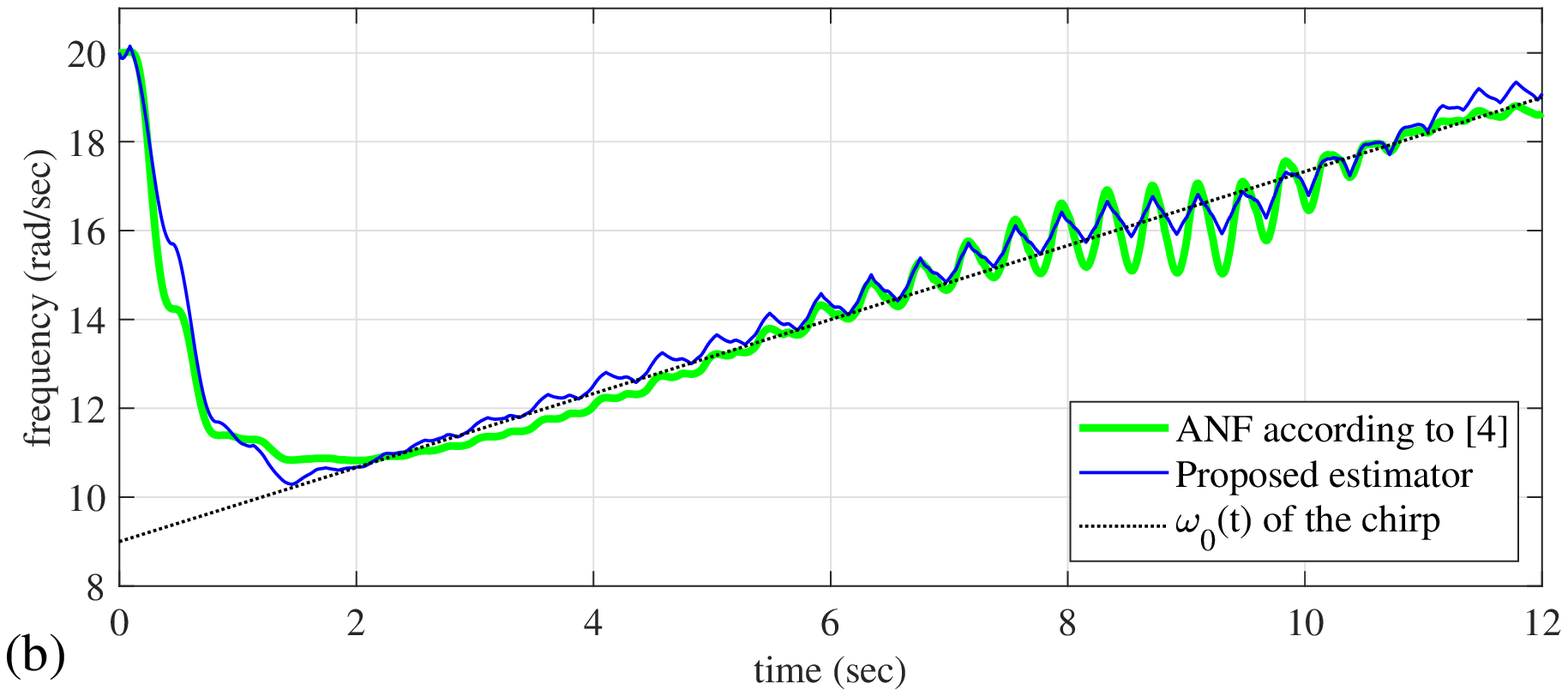}
\includegraphics[width=0.98\columnwidth]{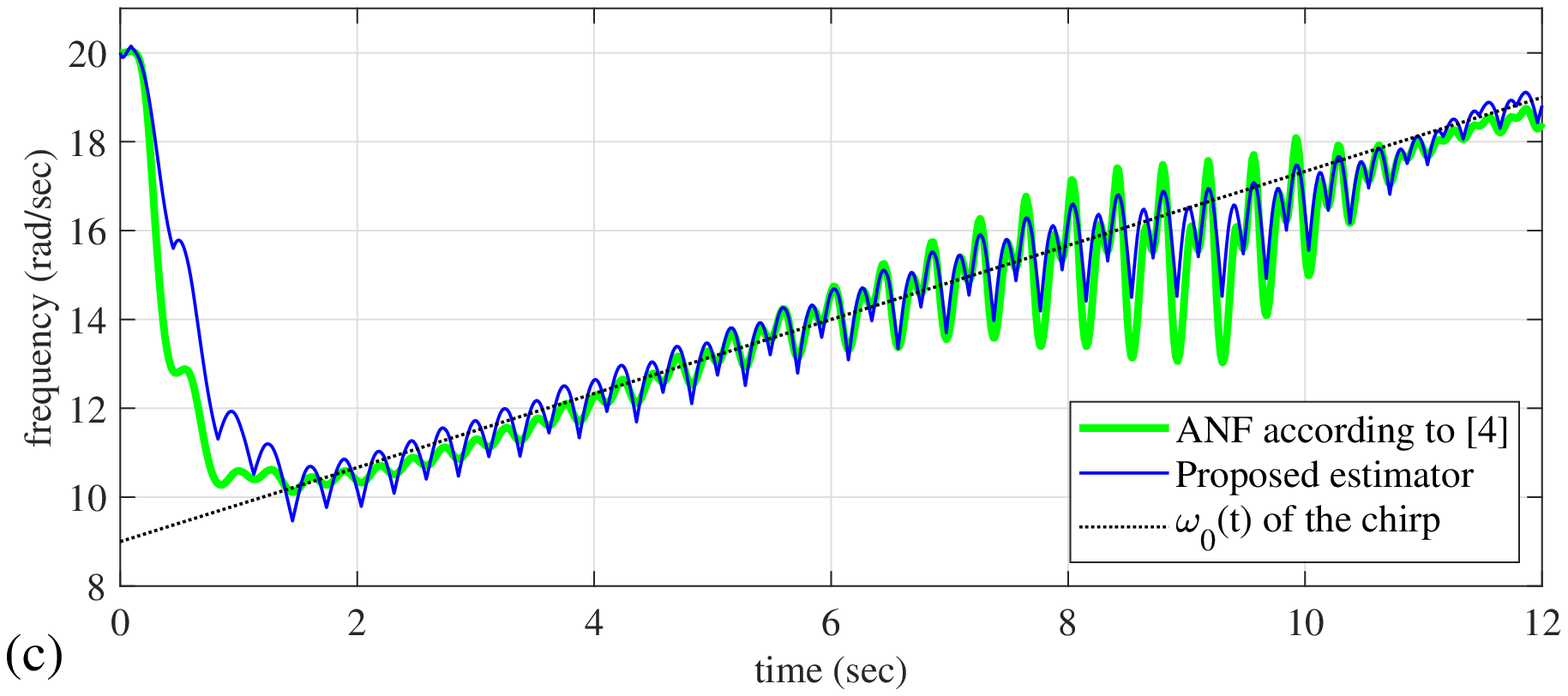}
\caption{Online estimate of the angular frequency $\theta(t)$
versus the chirp-driven $\omega_0(t)$: (a) for $\zeta=0.7$, (b)
$\zeta=1$, and (c) for $\zeta=1.3$.} \label{fig:expresults}
\end{figure}
One can recognize that in all three cases, the proposed estimation
algorithm follows the varying true angular frequency similarly as
\cite{mojiri2004} at the benefit of one design parameter. The
transient convergence appears slightly faster, and less deviations
to $\omega_0(t)$ appear for the critically damped case of
$\zeta=1$.

\section{Conclusions}
\label{sec:6}

In this paper, the problem of estimating the unknown frequency of
noisy sinusoidal signals with slowly varying amplitude has been
considered. The existing globally convergent frequency estimator
was modified by changing the scaling of the excitation signal and
output error, and canceling the damping ratio as a free design
parameter. Furthermore, the main robustification was achieved by
using the sign of an internal state, instead of the state itself,
within the adaptation law. Relying on the averaging theory of
periodic signals, an easy-to-follow and straightforward analysis
was developed in the frequency domain, assuming that the
timescales of a relatively fast harmonic (to be estimated) and
relatively slow drift of the amplitude can be separated. We
analyzed and proved the global asymptotic convergence of the
frequency estimate and determined the exponential convergence
rate. The dependency between band-limited white noise and the
resulting residual estimation error in a steady state was
established. The demonstrated numerical and experimental results
confirm the properties and performance of the proposed estimator.
A more detailed evaluation of the estimation performance and
comparison with other frequency estimation algorithms, also for
different experimental data-sets, are subject of our future works.

\bibliographystyle{elsarticle-num}        % Include this if you use bibtex

\bibliography{references}

\end{document}